\documentclass[12pt]{article}
\usepackage{graphicx}
\usepackage{amssymb,amsmath,amsfonts,graphics, setspace}
\usepackage{fullpage}
\usepackage{amsthm}

\usepackage{tikz}

\newtheorem{theorem}{Theorem}[section]

\newtheorem{proposition}{Proposition}[theorem]
\newtheorem{lemma}[theorem]{Lemma}
\newtheorem{example}[theorem]{Example}
\newtheorem{definition}[theorem]{Definition}

\newcommand{\xx}{\mathbf{x}}
\newcommand{\yy}{\mathbf{y}}
\newcommand{\uu}{\mathbf{u}}

\newcommand{\vv}{\mathbf{v}}

\newenvironment{frcseries}{\fontfamily{frc}\selectfont}{}

\usepackage[absolute,overlay]{textpos}

\newcommand{\fig}{\includegraphics}

\newcounter{comments}
\title{Dynamics Decomposition of Boolean Networks: An algebraic foundation}

\author{Alan Veliz-Cuba$^1$ 
\and
Claus Kadelka$^2$
\and
David Murrugarra$^3$
\and
Matthew Wheeler$^4$
\and
Kiara Boehringer$^1$
\and
Benjamin Coberly$^{2,5}$
\and
Reinhard Laubenbacher$^4$
}

\begin{document}

\maketitle

{\footnotesize
     \centerline{$^1$Department of Mathematics,
      University of Dayton, Dayton, OH 45469, USA}
}
{\footnotesize
     \centerline{$^2$Department of Mathematics, Iowa State University, Ames, IA 50011}
}
{\footnotesize
     \centerline{$^3$Department of Mathematics, University of Kentucky, 101 Main Building, Lexington, KY 40506}
}
{\footnotesize
     \centerline{$^4$Department of Medicine,
      University of Florida, Gainesville, FL 32610, USA}
}

{\footnotesize
     \centerline{$^5$Department of Computer Science, Iowa State University, Ames, IA 50011}
}

\begin{abstract}

Understanding the dynamics of Boolean networks is central to problems such as network reduction, design, control, and reverse engineering. As Boolean network models continue to grow in size and complexity, it becomes increasingly important to decompose networks into modules in a manner that is compatible with their dynamics. 
In this paper, we show that endowing the space of possible dynamics with a semiring structure enables a systematic decomposition of the dynamics of any Boolean network in terms of the dynamics of its constituent modules.
This algebraic framework provides a systematic way to analyze how local dynamical behaviors combine to produce global dynamics. Our results establish a concrete algebraic foundation for network modularity and introduce new mathematical tools for the study of complex Boolean networks, and opens the door to the application of algebraic methods to problems of network analysis, decomposition, and control.

\end{abstract}

\section{Introduction}
Since the beginnings of general systems theory in the 1950s, the relationship between structural features of a system and its dynamics has been of special interest. Systems are generally viewed as collections of many interacting components that generate system dynamics. In applications of systems theory, such as in systems biology or social science, insights into this relationship can be of practical interest, for instance for the purpose of structural interventions that are intended to modify the system's dynamics in desirable ways. The hypothesis of modularity, in particular, has been studied extensively from the beginning~\cite{hartwell1999molecular,wagner2007road}. Complex systems are assumed to be composed of subsystems called modules, and this modular structure is supposed to confer favorable traits on system dynamics, such as robustness to perturbations. 

In \cite{kadelka2023modularity}, we laid the foundation for a mathematical framework to study this notion in a rigorous fashion. We did this in the context of Boolean networks, a common modeling paradigm for applications to biology, engineering, physics, and computer science~\cite{schwab2020concepts,kadelka2024meta}. We defined a module of a Boolean network to be a maximal subset of nodes that are all connected through directed feedback loops, that is, directed loops in the dependency graph of the Boolean network variables. That is, a module is a strongly connected component in the dependency graph. The way these modules are connected can be seen in the partially ordered set called the condensation graph, which shows the dependency among strongly connected components.  Fig.~\ref{fig:modular_schematic} shows that each module can be seen as a network with potential inputs from upstream components and potential outputs that are fed to downstream components.  The main result in \cite{kadelka2023modularity} is that a decomposition of a Boolean network into the collection of strongly connected components induces a similar decomposition of the limit cycles in the state transition graph of the Boolean network. This result suffers from two shortcomings. First, it does not provide a modular decomposition of the entire state transition graph, but only the limit cycles. Second, the proof is a ``brute force'' verification of the decomposition, without the use of more advanced mathematical methods.

In this paper, we remedy both these shortcomings. We introduce a new mathematical framework to treat the state transition graph of a Boolean network. Instead of the customary representation as a directed graph with nodes all possible states of the network, and directed edges representing the transition from one state to another by applying the Boolean rules, we represent the dynamics as a collection of trajectories. Here, a trajectory is a directed path in the state transition graph that starts at a given initial state and ends when the trajectory becomes periodic. We then introduce algebraic operations on the set of trajectories, endowing it with the structure of a semiring. These algebraic operations extend previous approaches that were restricted to autonomous systems \cite{injectivitypoly2025,factorisationsemiring2024,rootsinsemiring2024}, making our framework applicable to non-autonomous networks. We then utilize the algebraic properties of this semiring to extend the decomposition result in \cite{kadelka2023modularity} from the limit cycles to the entire state transition graph.

In \cite{kadelka2023modularity}, the networks corresponding to strongly connected components in the wiring diagram were proposed as the modules of a Boolean network. The way these modules are connected can be seen in the partially ordered set called the condensation graph. This graph shows the dependency among strongly connected components.  Fig.~\ref{fig:modular_schematic} shows that each subnetwork corresponding to the strongly connected components, one can consider each module as a network with potential inputs from upstream components and potential outputs that are fed to downstream components.

\begin{figure}[h]
    \centering
    \fig[scale=0.6]{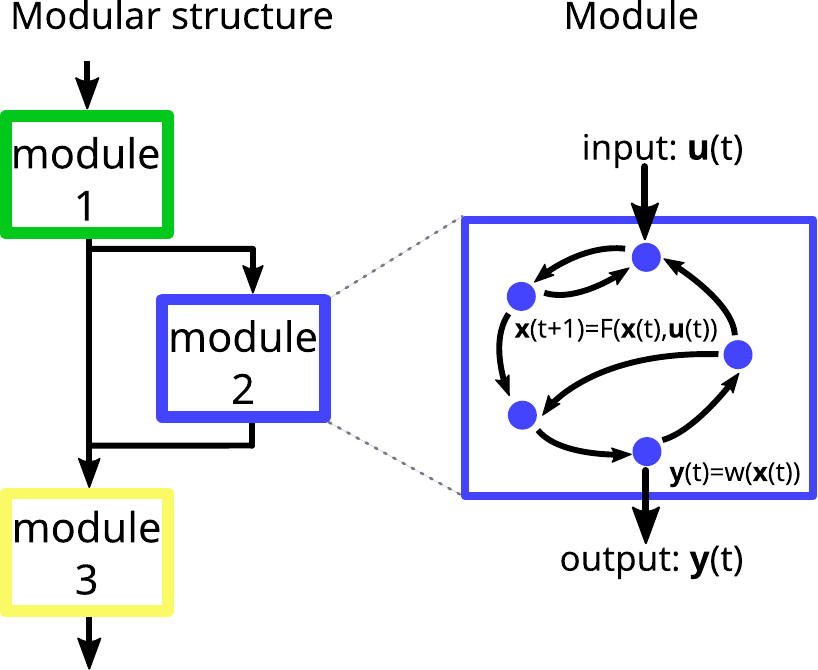}
    \caption{Schematic of modularity. Left: A collection of modules that interact. Right: Each module can be seen as a Boolean network that receives a time-dependent input from upstream modules, $\uu(t)$, evolves according to the rule $\xx(t+1)=F(\xx(t),\uu(t))$, and sends an output to downstream modules $\yy(t)=w(\xx(t))$. In practice $w$ is simply a projection on some of the variables of $\xx(t)$.
    }    \label{fig:modular_schematic}
\end{figure} 

In order to understand the dynamics of a network in terms of the dynamics of the modules, we need to have a standard representation of these, Fig.~\ref{fig:module}. 
For networks without inputs, the dynamics evolve according to the equation $\xx(t+1)=F(\xx(t))$ and are represented by the state transition graph which has $2^n$ vertices ($n$ is the number of variables). 
For networks with an external input $p$ that is fixed (representing e.g. an environmental condition), the dynamics evolve according to the equation $\xx(t+1)=F(\xx(t),p_0)$ and are represented by the state transition graph which has $2^n$ vertices (one such graph for each parameter value).
For networks with external input that is not fixed in time, the dynamics evolve according to the rule $\xx(t+1)=F(\xx(t),\uu(t))$ and there is no standard way  to represent the dynamics as a directed graph. 

\begin{figure}[h]
    \centering
    \fig[scale=0.5]{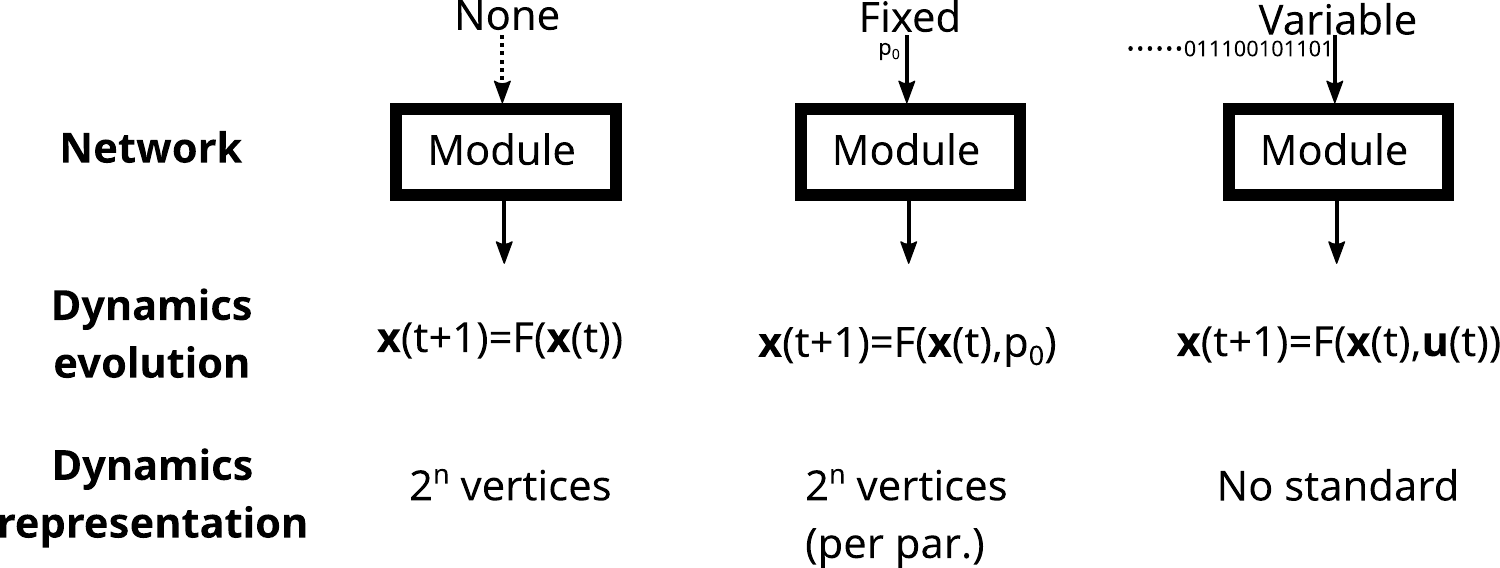}
    \caption{Different types of Boolean networks ($n$ variables). For networks with no inputs the dynamics is represented by the state transition graph which has $2^n$ vertices. For networks with an external fixed input $p$ the dynamics is represented by the state transition graph which has $2^n$ vertices (one such graph for each parameter value).
    For networks with variable external input there is no standard way used to represent the dynamics as a directed graph.
    }
    \label{fig:module}
\end{figure}

Our main contributions are the following. First, we formalize the concept of state transition graph for non-autonomous networks which generalize the standard state transition graph for typical Boolean networks (Sections  \ref{sec:sum}, \ref{sec:non-aut-intro}). Second, we present an algebraic framework that gives the space of possible dynamics the structure of a semiring (sum and product defined in Sections \ref{sec:sum}, \ref{sec:non-aut-intro}, and \ref{sec:prod_graphs}). Third, we use our framework to find a decomposition of the dynamics in terms of the dynamics of its modules (Theorem \ref{thm:main}). This ``factorization'' of dynamics provides a rigourous framework to study Boolean networks from a systems point of view, and in particular provides tools for the formalization of modularity.  We remark that we use  the framework of Boolean networks ($F:\{0,1\}^n\rightarrow\{0,1\}^n$) for clarity of presentation, but all results extend to finite dynamical systems ($F:X_1\times\ldots\times X_n\rightarrow X_1\times\ldots\times X_n$), including logical models and polynomial dynamical systems.

\section{Background and Preliminary Definitions}

\subsection{Dynamics of autonomous networks}\label{sec:dynamics}
\begin{definition}
A trajectory of a Boolean network $F:\{0,1\}^n\rightarrow \{0,1\}^n$ is a sequence $(\xx(t))_{t=0}^\infty$ of elements of $\{0,1\}^n$ such that $\xx(t+1)=F(\xx(t))$ for all $t\geq 0$.
\end{definition}

\begin{example}

Consider $F(x_1,x_2,x_3)=(x_2 \wedge \neg x_3,x_3,\neg x_1 \wedge x_2)$. There are 8 possible initializations for the trajectories (all binary strings with 3 elements), and they result in the following 8 trajectories (commas and parentheses for states are omitted for brevity; bold text indicates the first instance of the periodic part).
\end{example}

 $T_{000}=(\textbf{000},000,000,000,000,\ldots)$
 
 $T_{001}=(001,\textbf{010},\textbf{101},010,101,\ldots)$
 
 $T_{010}=(\textbf{010},\textbf{101},010,101,010,\ldots)$ 
 
 $T_{011}=(\textbf{011},011,011,011,011,\ldots)$
 
 $T_{100}=(100,\textbf{000},000,000,000,\ldots)$
 
 $T_{101}=(\textbf{101},\textbf{010},101,010,101,\ldots)$
 
 $T_{110}=(110,100,\textbf{000},000,000,\ldots)$ 
 
 $T_{111}=(111,\textbf{010},\textbf{101}, 010 , 101,\ldots)$

We can see that $T_{010}$ and $T_{101}$ are periodic trajectories with period 2. Similarly, $T_{000}$ and $T_{011}$ are periodic and have period 1. These trajectories are shown in Fig.~\ref{fig:dynamics_eg}b. A more compact way to encode the trajectories is shown in Fig.~\ref{fig:dynamics_eg}c, where the periodic pattern has been encoded as a cycle. Note that the initial states are marked by rectangles. Although such markers may be omitted for some trajectories, they are necessary for $T_{010}$ and $T_{101}$ to specify the initial state. Using an appropriate quotient of graphs we can assume that the graphs in Fig.~\ref{fig:dynamics_eg}c \textbf{are} the trajectories in Fig.~\ref{fig:dynamics_eg}b.

For autonomous systems, there is a compact way to represent the dynamics and obtain the ``state transition graph'',  that is some times called  the state space, Fig.~\ref{fig:dynamics_eg}d

\begin{definition}\label{def:stg_classic}
The \emph{state transition graph} (STG) of a Boolean network $F:\{0,1\}^n\rightarrow \{0,1\}^n$ is a directed graph with vertices in $\{0,1\}^n$ and an edge from $\xx$ to $\yy$ if $F(\xx)=\yy$. Unless specified otherwise, by graphs, we refer to directed graphs. 
\end{definition}

\begin{example}
Consider again the Boolean network given by 
$F(x_1,x_2,x_3)=(x_2 \wedge \neg x_3,x_3,\neg x_1 \wedge x_2)$. Its state transition graph is shown in Fig.~\ref{fig:dynamics_eg}b.
\end{example}

\begin{definition}\label{def_attractors}
An \emph{attractor of length $r$} is a set with $r$ elements, $\mathcal{C}=\{c_1,\ldots,c_r\}$, such that $F(c_1)=c_2, F(c_2)=c_3,\ldots, F(c_{r-1})=c_r, F(c_r)=c_1$.
\end{definition}

\begin{example}\label{eg:wd_dyn}
The network from the previous example, $F(x_1,x_2,x_3)=(x_2 \wedge \neg x_3,x_3,\neg x_1 \wedge x_2)$, has 2 fixed points (i.e., attractors of length 1) and one cycle of length 2, which can be easily seen in the state transition graph representation (Fig.~\ref{fig:dynamics_eg}d).
\end{example}

\begin{figure}[h]
    \centering
    \fig[scale=1.5]{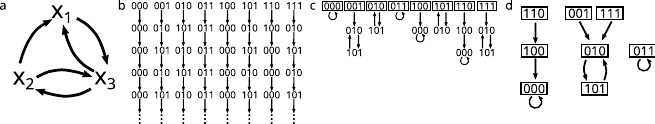}
    \caption{Wiring diagram and state transition graph of the Boolean network in Example~\ref{eg:wd_dyn}. 
    (a) The wiring diagram encodes the dependency between variables. The edges $x_3\rightarrow x_1$ and $x_1\rightarrow x_3$ are negative because $f_1=x_2 \wedge \neg x_3$ is decreasing in $x_3$ and $f_3=\neg x_1\wedge x_2$ is decreasing in $x_1$.
    (b) Trajectories represented as paths.
    (c) By using cycles, we can represent the trajectories as finite graphs. Note that we must indicate which state is the initial state to avoid ambiguity (e.g. the third and sixth graphs are not the same trajectory). Here and throughout the paper, states within rectangles denote initial states. By an appropriate quotient of graphs in (b) we can say the graphs in (c) \textbf{are} the trajectories.  
    (d) The state transition graph is a graph with edges between inputs and outputs.
    }
    \label{fig:dynamics_eg}
\end{figure}

\subsection{Sum of trajectories}\label{sec:sum}

In the rest of the section, by a trajectory $T_ {\xx,F}$ we mean either the sequence $(\xx,F(\xx),\ldots)$ (e.g. Fig.~\ref{fig:dynamics_eg}b) or its compact representation (e.g. Fig.~\ref{fig:dynamics_eg}c). Since they are the same up to a graph quotient, there will be no ambiguity. Also, whenever there is no ambiguity, we will denote $T_ {\xx,F}$ simply by $T_ {\xx}$. We now discuss another way to represent the trajectories in Fig.~\ref{fig:dynamics_eg}c that will motivate an operation between graphs.

We first look at $T_{000}$ and $T_{110}$  in Fig.~\ref{fig:dynamics_eg}bc. We can see that $T_{110}$ contains the information needed to encode $T_{000}$ (the fact that 000 transitions to itself). Note that here we are implicitly using the fact that the network is autonomous; since 000 transitions to itself in the third step of $T_{110}$, 000 transitions to itself regardless of the step number. In fact, we can say that $T_{110}$ actually ``contains'' $T_{000}$ in Fig.~\ref{fig:dynamics_eg}c (including the initial-state label, the rectangle). A graph that can encode both is shown in Fig.~\ref{fig:sum_intuition}a. Note that we keep the marker that specifies which vertices are initial states. Now we look at $T_{001}$  and $T_{111}$ in Fig.~\ref{fig:dynamics_eg}bc. We can see that there is an ``overlap'' and hence we can use a combined graph that encodes both. This process of combining trajectories to obtain a minimal representation will result in the graph in  Fig.~\ref{fig:sum_intuition}b. 

\begin{figure}[h]
    \centering
    \fig[scale=1.5]{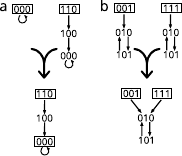}
    \caption{Representing more than one trajectory by a single graph. 
    (a) Since the first graph is contained in the second, we can simply indicate (by putting a rectangle around 000) that 000 is also a valid initial state.
    (b) Since the graphs overlap, we can merge the overlap into a single instance.
    }
    \label{fig:sum_intuition}
\end{figure}

This combination of graphs motivates the following formal definitions.
\begin{definition}\label{def:graph}
In the rest of the manuscript, by a \emph{graph} we mean $G=(V,E,I)$, where $V\subseteq \{0,1\}^n$ is the set of states, $E\subseteq V\times V$ is the set of (directed) edges, and $I\subseteq V$ is the set of initial states. We also say that a graph $G'=(V',E',I')$ is contained in $G$ if $V'\subseteq V$, $E'\subseteq E$, and $I'\subseteq I$.
\end{definition}

\begin{definition}\label{def:sum_basic}
Consider the graphs $G_1=(V_1,E_1,I_1)$ and $G_2=(V_2,E_2,I_2)$. The \emph{sum} of $G_1$ and $G_2$ is the minimal graph that contains both $G_1$ and $G_2$, that is, $G_1+G_2:=(V_1\cup V_2,E_1\cup E_2, I_1\cup I_2)$. If there is no potential for ambiguity, the plus symbol will be omitted in the figures, particularly in the case where the sum has disconnected components.
\end{definition}

Fig. \ref{fig:sum_example} shows examples of the sum of some of the graphs from Fig. \ref{fig:sum_intuition}c. The first property of this operation is given by the following proposition.

\begin{figure}[h]
    \centering
    \fig[scale=1.5]{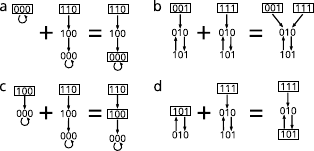}
    \caption{Examples of sums of graphs.
    }
    \label{fig:sum_example}
\end{figure} 

\begin{proposition} The sum of graphs in Definition \ref{def:sum_basic} is associative and commutative.
\end{proposition}

\begin{proof}
The proof follows from the fact that addition is defined using unions of sets.
\end{proof}

With this definition, we now formally provide an alternative definition of the dynamics or the state transition graph of a (synchronously updated) Boolean network, which we will later generalize for non-autonomous networks. 
\begin{definition}\label{def:alt_def_stg}
    The state transition graph (or dynamics) of a Boolean network $F$ is $\displaystyle \mathcal{D}(F) = 
    \sum_{ \xx \in \{0,1\}^n } T_{\xx}$, where $T_{\xx}$ denotes the trajectory of $F$ that starts at $\xx$.
\end{definition}

\begin{example}
Consider again the Boolean network given by 
$F(x_1,x_2,x_3)=(x_2 \wedge \neg x_3,x_3,\neg x_1 \wedge x_2)$. Its state transition graph as a sum of trajectories is shown in Fig.~\ref{fig:sum_autonomous}b.
\end{example}

\begin{figure}[h]
    \centering
    \fig[scale=1.5]{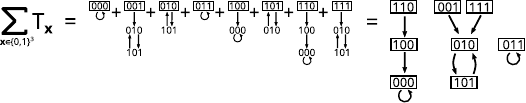}
    \caption{Example of a state transition graph as the formal sum of trajectories. Note that we are omitting the plus symbol for brevity. The resulting graph coincides with the state transition graph in  Fig.~\ref{fig:dynamics_eg}d.
    }
    \label{fig:sum_autonomous}
\end{figure} 

Although the resulting graph in Definition \ref{def:sum_basic} can be seen as the typical state transition graph, for non autonomous networks it is no longer the case as we will see in the next section.

\subsection{Dynamics of non-autonomous networks}
\label{sec:non-aut-intro}

\begin{definition}
A \emph{non-autonomous Boolean network} is defined by $$\xx(t+1)=F(\xx(t),\uu(t)),$$
where $F:\{0,1\}^{m+n}\rightarrow \{0,1\}^n$ and $\left(\uu(t)\right)_{t=0}^\infty$ is a sequence with elements in $\{0,1\}^m$, which we denote with $\uu$. We call this type of network non-autonomous because the dynamics will depend on the particular sequence $\uu$ that we consider (and hence, on time). We use $F^\uu$ to denote this non-autonomous network given by $F$ and $\uu$. If there is no ambiguity, we will denote the non-autonomous network simply by $F$.
\end{definition}

\begin{example}\label{eg:na_example1}
Consider the non-autonomous network given by $F(\xx,\uu)=(u_1\wedge x_2, x_1)$ ($F:\{0,1\}^{2+1}\rightarrow \{0,1\}^2$), and the sequence $\uu=(1,1,0,1,0,1,0,1,\ldots)$, with wiring diagram in Fig.~\ref{fig:na_dynamics_eg1}a. The trajectories of this network are the following.

$T_{00}=(\mathbf{00},00,00,00,00,00,00,\ldots)$

$T_{01}=(01,10,01,\mathbf{00},00,00,00,\ldots)$

$T_{10}=(\mathbf{10},\mathbf{01},10,01,10,01,10,\ldots)$

$T_{11}=(11,11,11,\mathbf{01},\mathbf{10},01,10,\ldots)$,

where the bold text indicates the first instance of the periodic part. We can see that regardless of the initial state, eventually the trajectory achieves periodicity (period 1 or 2).

It is important to note that the typical state transition graph representation given in Definition \ref{def:stg_classic} will simply not work. For instance, for an autonomous network, the state 01 has a unique successor, but for the non-autonomous network we considered, we can see that 01 has 00 as a successor in $T_{01}$ (for $t=2$) and 10 as a successor in $T_{10}$. Thus, the standard definition of state transition graph is meaningless for non-autonomous networks.

To illustrate our approach to solve this issue, we first show a  graphical representation of the trajectories in Fig.~\ref{fig:na_dynamics_eg1}b that will motivate a generalization of Definition \ref{def:alt_def_stg}. 
First, for each trajectory in Fig.~\ref{fig:na_dynamics_eg1}b we consider the compact representation using cycles for the periodic part Fig.~\ref{fig:na_dynamics_eg1}c. Note that the resulting graph can have repeated vertices. Formally, we are using the following definition. 

\end{example}

\begin{definition}
Consider a trajectory $T_\xx$ of the form $(\vv_1,\ldots,\vv_l,\vv_{l+1},\ldots,\vv_{l+k},\ldots)$, where $(\vv_1,\ldots,\vv_l)$ is the transient part and $(\vv_{l+1},\ldots,\vv_{l+k})$ is the periodic part. Its ``compact representation'' is the graph given by $\vv_1\rightarrow \vv_2 \rightarrow \cdots \rightarrow \vv_{l+k}$ and an edge from $\vv_{l+k}$ to $\vv_{l+1}$, where $\vv_1$ is also given an ``initial state marker'' (a rectangle in our figures). Since a trajectory has a unique compact representation and is completely determined by it, we will also refer to the  compact representation as \textbf{the} trajectory $T_\xx$. 
\end{definition}

\begin{example}\label{eg:na_dynamics_eg1}
From Example \ref{eg:na_example1}, consider the trajectory $T_{01}=(01,10,01,\mathbf{00},00,00,00,\ldots)$ (shown graphically in Fig.~\ref{fig:na_dynamics_eg1}b, second path). Then, the transient part is $(01,10,01)$ and the periodic part is $(00)$. Hence, the compact representation of this trajectory is $\fbox{$01$}\rightarrow 10\rightarrow 01\rightarrow 00$ and an edge from $00$ to itself (Fig.~\ref{fig:na_dynamics_eg1}c, second graph). In contrast to Section~\ref{sec:dynamics}, we see that the same vertex can appear twice in the graph (01 in this case). 

Now, consider the trajectory $T_{10}=(\mathbf{10},\mathbf{01},10,01,10,01,10,\ldots)$ (shown graphically in Fig.~\ref{fig:na_dynamics_eg1}b, third path). Then, the transient part is empty and the periodic part is $(10,01)$. Hence, the compact representation of this trajectory is $\fbox{$10$}\rightleftarrows 01$ (Fig.~\ref{fig:na_dynamics_eg1}c, third graph). Since the graph is a cycle, here we explicitly observe the need to specify which vertex is the initial state. Without the initial state marker we would not know if the graph represents a trajectory that starts at 10 or 01.

Consider the trajectory $T_{11}=(11,11,11,\mathbf{01},\mathbf{10},01,10,\ldots)$ (shown graphically in Fig.~\ref{fig:na_dynamics_eg1}b, fourth path). Then, the transient part is $(11,11,11)$ and the periodic part is $(01,10)$. Hence, the compact representation of this trajectory is $\fbox{$11$}\rightarrow 11\rightarrow 11\rightarrow 01 \rightleftarrows 10$ (Fig.~\ref{fig:na_dynamics_eg1}c, fourth graph). Again, we observe that vertices can appear more than once.
\end{example}

Now, we will describe a sum of graphs first with examples and then formally that will generalize Definition \ref{def:sum_basic} and will be used to generalize the definition of state transition graph to non autonomous Boolean networks.

Consider trajectories $T_{00}$ and $T_{01}$ in Fig.~\ref{fig:na_dynamics_eg1}c. We can see that the first trajectory is ``contained'' in the second trajectory in the sense that the path obtained by starting at $00$ in $T_{00}$ is the same as the path obtained by starting at $00$ in $T_{01}$. However, we must remember that since the network is not autonomous, such property is not always true for other vertices. Indeed, the path obtained by starting at the second instance of $01$ in $T_{01}$ is not the path obtained by starting at  the first instance of $01$ in $T_{01}$. Also, for example, the path obtained by starting at $10$ in $T_{01}$ is not equal to the path obtained by starting at $10$ in $T_{10}$. Hence, we must somehow indicate that in $T_{01}$ it valid to start at $00$. We can achieve this easily by simply marking 00 as an initial state, Fig.~\ref{fig:na_dynamics_eg1}d (left). Similarly, consider the trajectories $T_{10}$ and $T_{11}$  in Fig.~\ref{fig:na_dynamics_eg1}c. Since $T_{10}$ is ``contained'' in $T_{11}$, it is enough to mark $10$ in $T_{11}$ with the initial state marker, Fig.~\ref{fig:na_dynamics_eg1}d (right). This example again shows the importance of the initial state marker in determining which vertices are valid initial states. Since there is no overlap between the two graphs in Fig.~\ref{fig:na_dynamics_eg1}d, the process is done.

\begin{figure}[h]
    \centering
   \fig[scale=1.5]{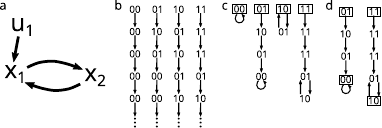}
    \caption{Wiring diagram and state transition graph of the non-autonomous Boolean network in Example~\ref{eg:na_example1}. 
    (a) The wiring diagram encodes the dependency between variables. Note that the fact that the first entry of $F$ depends on $u_1$ is encoded graphically as an arrow from $u_1$ to $x_1$.
    (b) The four trajectories of the non-autonomous BN.
    (c) Periodic points in each trajectory can be represented more compactly.
    (d) We can encode (c) more compactly by ``merging'' common subgraphs between trajectories. The merging has to be done in a way that the non-autonomous properties are maintained. 
    }
    \label{fig:na_dynamics_eg1}
\end{figure} 

At its core, the sum operation will be  conceptualized as a structural overlay of state transitions. Directly related to the intuition of overlay is the concept of containment, which in our context specifically means containment that respects trajectories ---a notion that can differ from typical containment in graph theory. For instance, the graph
\[G = \fbox{$11$}\rightarrow 
\begin{tikzpicture}[baseline=(X.base)]
    \node[inner sep=1pt] (X) {11};
    \path[->, every loop/.style={looseness=2}] (X) edge [loop above] (X);
\end{tikzpicture} 
\rightarrow 01 \rightleftarrows 10\]
might, under definitions based on label mapping or graph homomorphisms, be considered to contain the graph 
\[T_{11}=\fbox{$11$}\rightarrow 11\rightarrow 11\rightarrow 01 \rightleftarrows 10\]
because every vertex and edge of $T_{11}$ exists in $G$. However, graph $G$ fails in encoding the fact that only 3 consecutive $11$ states can occur, and in that sense it does not contain all trajectory information of $T_{11}$. To keep the underlying principles focused and avoid the unnecessary overhead of multisets or the explicit ``unrolling'' of the graph over time, we define containment structurally: We say that graph $G_1$ is contained in $G_2$ if $G_1$ can be obtained from $G_2$ by removing edges, vertices, and initial state markers. Consequently, $T_{11}$ is not contained in $G$, whereas  it is contained in 
\[\fbox{$11$}\rightarrow 11\rightarrow 11\rightarrow 01 \rightleftarrows \fbox{10}.\]

With this clarification of containment, the formal definition of the sum (Definition~\ref{def:sum_basic}) and its properties apply verbatim to non-autonomous dynamics, Fig.~\ref{fig:na_sum_example}. Furthermore, to formally define the state transition graph of non-autonomous networks, we can directly use Definition~\ref{def:alt_def_stg}, and hence \emph{the state transition graph of a network is the sum of its trajectories in a precise mathematical sense.}

\begin{figure}[h]
    \centering
    \fig[scale=1.5]{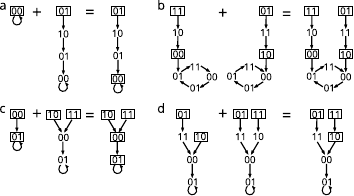}
    \caption{Examples of sums of graphs.
    }
    \label{fig:na_sum_example}
\end{figure} 

\begin{example}\label{eg:na_example2}
Consider the non-autonomous network given by $F(\xx,\uu)=(u_1 \wedge x_2, \neg x_1)$, and the sequence $\uu=(1,1,0,1,0,1,0,1,\ldots)$, with wiring diagram in Fig.~\ref{fig:na_dynamics_eg2}a. The trajectories of this network are
\begin{align*}T_{00}&=(00,\mathbf{01,11,00,01},01,11,00,01,01,11,00,01,01,11,00,...),\\
T_{01}&=(01,11,10,\mathbf{00,01,01,11},00,01,01,11,00,01,01,11,00,...),\\
T_{10}&=(10,\mathbf{00,01,01,11},00,01,01,11,00,01,01,11,00,01,01,...),\\
T_{11}&=(11,10,00,\mathbf{01,11,00,01},01,11,00,01,01,11,00,01,01,...).
\end{align*}

The trajectories are also shown as graphs in Fig.~\ref{fig:na_dynamics_eg2}b. The sum to obtain the state transition graph is shown in Fig.~\ref{fig:na_dynamics_eg2}cd. Note that some states appear multiple times because the network is non-autonomous.

\end{example}

\begin{figure}[h]
    \centering
   \fig[scale=1.5]{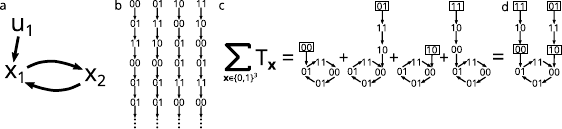}
    \caption{Wiring diagram and state transition graph of the non-autonomous Boolean network in Example~\ref{eg:na_example2}. 
    (a) The wiring diagram encodes the dependency between variables. Note that the fact that the first entry of $F$ depends on $u_1$ is encoded graphically as an arrow from $u_1$ to $x_1$.
    (b) The four trajectories of the non-autonomous BN.
    (c) Periodic points in each trajectory can be represented more compactly and are the terms in the sum that defines the state transition graph.
    (d) The resulting sum is the state transition graph. 
    }
    \label{fig:na_dynamics_eg2}
\end{figure}

The dependence of the dynamics on the input sequence becomes clear when considering the same non-autonomous network but a different input sequence.
  
\begin{example}\label{eg:na_example3}
Let $F(\xx,\uu)=(u_1\wedge x_2, \neg x_1)$ with $\mathbf u=(0,0,0,0,0,\ldots)$. Then, the wiring diagram, shown in Fig.~\ref{fig:na_dynamics_eg3}a, is exactly as in Fig.~\ref{fig:na_dynamics_eg2}a. However, the trajectories
\begin{align*}
T_{00}=(00,\mathbf{01},01,01,01,...),\\
T_{01}=(\mathbf{01},01,01,01,01,...),\\
T_{10}=(10,00,\mathbf{01},01,01,...),\\
T_{11}=(11,00,\mathbf{01},01,01,...),
\end{align*}
which are shown as graphs in Fig.~\ref{fig:na_dynamics_eg3}b, and the sum to obtain the state transition diagram (Fig.~\ref{fig:na_dynamics_eg3}c,d)  differ substantially.

\end{example}

\begin{figure}[h]
    \centering
   \fig[scale=1.5]{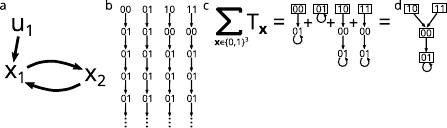}
    \caption{Wiring diagram and state transition graph of the non-autonomous Boolean network in Example~\ref{eg:na_example3}. 
    (a) The wiring diagram encodes the dependency between variables. Note that the fact that the first entry of $F$ depends on $u_1$ is encoded graphically as an arrow from $u_1$ to $x_1$.
    (b) The four trajectories of the non-autonomous BN.
    (c) Periodic points in each trajectory can be represented more compactly and are the terms in the sum that defines the state transition graph.
    (d) The resulting sum is the state transition graph. 
    }
    \label{fig:na_dynamics_eg3}
\end{figure}

\section{Product of dynamics}\label{sec:prod}

In this section, we define the product of (non-autonomous) dynamics. 
We begin by introducing a running example.

\begin{example}\label{ex:run}
Consider the Boolean network $F(x_1,x_2,x_3,x_4) = 
(\neg x_1 \vee \neg x_2, \neg x_1 \vee x_2, x_4\wedge x_1 , x_3)$, with  wiring diagram shown in Fig.~\ref{fig:running_eg}a.
We can see this Boolean network as the combination of two modules, $F_1$ and $F_2$, Fig.~\ref{fig:running_eg}bc. The top module $F_1(x_1,x_2)=(\neg x_1 \vee \neg x_2, \neg x_1 \vee x_2)$ has no external input (Fig.~\ref{fig:running_eg}b). Thus, its dynamics results in the typical state transition graph with $2^2=4$ elements and each state is a valid initial state, Fig.~\ref{fig:running_eg}d. The node $x_1$ of $F_1$ acts as the external input $u_1$ of the bottom module, which is thus a non-autonomous network $F_2(x_3,x_4,u_1)=(x_4\wedge u_1, x_3)$, with wiring diagram in Fig.~\ref{fig:running_eg}c.
To enumerate all trajectories of $F_2$, one must consider all possible  external input sequences of $F_2$, which are determined by $x_1$. Depending on the choice of initial state of $F_1$, the trajectories of $x_1$ can be 0101..., 1111..., and 1010.... The dynamics of $F_2$, given these three input sequences $u_1(t)$, are shown in Fig.~\ref{fig:running_eg}e.
\end{example}

\begin{figure}[h]
\centerline{ \hbox{ 
 \fig[scale=1.5]{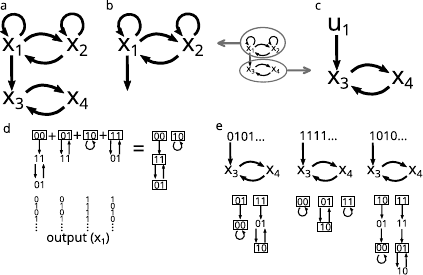}
}}
 \caption{ Dynamics of modules. 
 (a) Wiring diagram of full network. 
 (b) Top module.
 (c) Dynamics of the top module with the resulting output shown at the bottom.
 (d) Bottom module.
 (e) Dynamics of the bottom module corresponding to the different possible inputs it can receive from upstream. Note that since the output sequence 0101... appears twice in (c), there are exactly three possible input sequences for the bottom module.
}
 \label{fig:running_eg}
\end{figure}

\subsection{Product of graphs}\label{sec:prod_graphs}

The tools we introduced in Section 2 permit the computation of the dynamics of different modules separately. However, of key interest from a systems point of view is the problem of reconstructing the dynamics of the full network from the dynamics of the modules. To be able to do this, we need to define a product on graphs.

\begin{definition}\label{def:eq_reach}
Consider the graphs $G_1=(V_1,E_1,I_1)$ and $G_2=(V_2,E_2,I_2)$.  We say that $a\in V_1$ and $b\in V_2$ are \emph{equally reachable} if there are initial states $x\in I_1$ and $y\in I_2$ such that there are paths of the same length from $x$ to $a$ and from $y$ to $b$. This definition will allow us to quantify the idea of timing being consistent for states in state transition graphs of different modules to define the product of dynamics properly. Note that we allow a path of length zero, so for any $a\in I_1$ and $b\in I_2$, $a$ and $b$ are equally reachable.
\end{definition}

To illustrate this definition, consider the pairs of graphs shown in Fig.~\ref{fig:running_eg_eq_reach}. 
The pairs of states that are equally reachable are shown below each of the graphs (we omit commas and parentheses for simplicity). For example, in Fig.~\ref{fig:running_eg_eq_reach}, every combination of initial states from $G_1$ and $G_2$ will be equally reachable. Also, in Fig.~\ref{fig:running_eg_eq_reach}a, $11$ in $G_1$ and $01$ in $G_2$ are equally reachable because they can be reached from $00\in I_1$ and $11\in I_2$ (after one time step), respectively. 
In Fig.~\ref{fig:running_eg_eq_reach}b, $01$ in $G_1$ and $01$ in $G_2$ are equally reachable because they can be reached from $11\in I_1$ and $10\in I_2$ (after one time step), respectively; also, $01$ and $00$ are equally reachable because they can be reached from $11\in I_1$ and $00\in I_2$ (after one time step), respectively. 

\begin{figure}[h]
\centerline{ \hbox{ 
 \fig[scale=1.5]{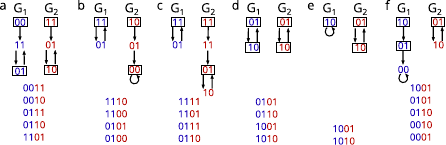}
}}
 \caption{ Illustration of equal reachability. The states at the bottom (color coded) are the concatenation of states of $G_1$ (in blue) and states of $G_2$ (in red) that are equally reachable.
}
 \label{fig:running_eg_eq_reach}
\end{figure}

With the concept of equal reachability, we can now define the product of graphs.

\begin{definition}\label{def:prod}
Consider the graphs $G_1=(V_1,E_1,I_1)$ and $G_2=(V_2,E_2,I_2)$. We define the product of $G_1$ and $G_2$ as the graph $G=(V,E,I)$ such that 

\[
V = \{ (a,b)\in V_1\times V_2 :  \text{ $a$ and $b$ are equally reachable } \},
\]
there is an edge in $E$ from $(a,b)\in V$ to $(c,d)\in V$ if $(a,c)\in E_1$ and $(b,d)\in E_2$, and $I=I_1\times I_2$ (cross product).

\end{definition}

\begin{example}\label{ex:prod} Consider $G_1$ and $G_2$ from Fig.~\ref{fig:running_eg_eq_reach}a. Using Definitions~\ref{def:eq_reach} and \ref{def:prod}, we obtain

\[
V=\{ 0011,0010,0111,0110,1101 \}, \ I=\{ 0011,0010,0111,0110\}, 
\text{ and} 
\]

\[
E= \{ (0011,1101) , (0010,1101), (0111, 1101), (0110,1101),(1101,0110)
 \},
\]

where for ease of visualization we are denoting $(x,y)$ simply by $xy$ (concatenation) in $V$, $I$, and $E$. The resulting product is shown in Fig.~\ref{fig:running_eg_prod}a.

Now consider $G_1$ and $G_2$ from Fig.~\ref{fig:running_eg_eq_reach}c. Using Definitions~\ref{def:eq_reach} and \ref{def:prod} we obtain

\[
V=\{ 1111,1101,0111,0110 \}, \ I=\{ 1111,1101\}, 
\text{ and} 
\]

\[
E= \{ (1111,0111), (1101,0110), (0111,1101),(0110,1101)
 \}.
\]
The resulting product is shown in Fig.~\ref{fig:running_eg_prod}c.
\end{example}

\begin{figure}[h]
\centerline{ \hbox{ 
 \fig[scale=1.5]{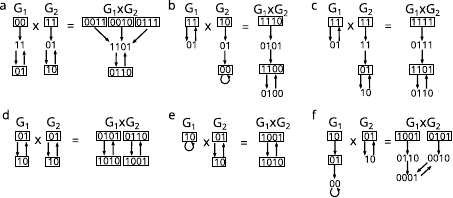}
}}
 \caption{ Examples of the product of graphs. These examples correspond to the graphs in Fig.~\ref{fig:running_eg_eq_reach}.
}
 \label{fig:running_eg_prod}
\end{figure}

With the definitions of addition and product, we have the following properties.
\begin{proposition}
Consider the graphs $G_1, G_2, G_3$, then (up to isomorphisms) 
\[
(G_1+G_2)\times G_3 = G_1\times G_3 + G_2\times G_3 \ , \ \ \ 
G_1\times (G_2 + G_3) = G_1\times G_2 + G_1\times G_3,
\]
\[
 \text{ and } \ \ \ 
(G_1\times G_2)\times G_3 = G_1\times (G_2\times G_3).
\]

In particular, these operations give the space of possible dynamics the structure of a semiring with the empty set as the zero element. If required by definition, a multiplicative identity can be adjoined in the standard way.

\end{proposition}

\subsection{Main result}\label{sec:main result}

To be able to relate the dynamics of a network with the dynamics of its modules, we will need the following lemma, which we state and prove in a more general setting, where two Boolean subnetworks can both influence each other. Subsequently, we will assume that the influence is at most one-directional, i.e., that the network contains two or more strongly connected components.

\begin{lemma}\label{lemma:prod}
Consider a Boolean network $F:\{0,1\}^{m+n}\rightarrow \{0,1\}^{m+n}$ 
such that it has the form $F(\xx_1,\xx_2)=(F_1(\xx_1,\xx_2),F_2(\xx_1,\xx_2))$ where $F_1:\{0,1\}^{n+m}
\rightarrow \{0,1\}^n$ and $F_2:\{0,1\}^{n+m}\rightarrow \{0,1\}^m$. Consider $\xx=(\xx_1,\xx_2)$ where $\xx_1\in \{0,1\}^n$ and $\xx_2\in \{0,1\}^m$ and denote with $T_{\xx,F}$ the trajectory of $F$ that starts at $\xx$, denote with $T_{\xx_1,F_1}$ the trajectory of $F_1$ that starts at $\xx_1$ and has external input $\xx_2$, and denote with $T_{\xx_2,F_2}$ the trajectory of $F_2$ that starts at $\xx_2$ and has external input $\xx_1$.
Then, $T_{\xx,F} =T_{(\xx_1,\xx_2),F}= 
T_{\xx_1,F_1}  \times T_{\xx_2,F_2} $. 
\end{lemma}

\begin{proof} First, since the graphs $T_{\xx_1,F_1}$ and $T_{\xx_2,F_2}$ have a unique initial state marker ($\xx_1$ and $\xx_2$, respectively), then the product will have a unique initial state marker. Namely, the only initial state in $T_{\xx_1,F_1} \times T_{\xx_2,F_2}$ is $(\xx_1,\xx_2)=\xx$. 

Since there is an edge from $\xx_1$ to $F_1(\xx_1,\xx_2)$ in $T_{\xx_1,F_1}$ and from $\xx_2$ to $F_2(\xx_1,\xx_2)$ in $T_{\xx_2,F_2}$, there will be an edge from $\xx=(\xx_1,\xx_2)$ to $(F_1(\xx_1,\xx_2),F_2(\xx_1,\xx_2))=F(\xx)$ in $T_{\xx_1,F_1} \times T_{\xx_2,F_2} $. That is, both $T_{\xx_1,F_1} \times T_{\xx_2,F_2} $ and $T_{\xx,F}$ coincide in the initial state marker and the path after one time step. Similarly, it follows that $T_{\xx_1,F_1} \times T_{\xx_2,F_2} $ and $T_{\xx,F}$ coincide at every path; hence they are equal and that finishes the proof.
\end{proof}

Now, suppose $F$ is of the form $ F(\xx_1,\xx_2)=(F_1(\xx_1),F_2(\xx_1,\xx_2))$. Then, $F_1$ can be considered an autonomous BN and $F_2$ can be considered a non-autonmous BN with external inputs given by $T_{\xx_1,F_1}$. In isolation, $F_2$ can have a variety of external inputs and hence a variety of possible dynamics. We denote the family of all possible dynamics of $F_2$ by $\mathcal{D}(F_2)$ and define the semidirect product of $\mathcal{D}(F_1)$ and $\mathcal{D}(F_2)$ as the sum of compatible products of dynamics of $F_1$ and $F_2$, that is, the sum of the product of trajectories of $F_1$, $T_{\xx_1,F_1}$, times a compatible dynamics of $F_2$, $\mathcal{D}(F_2^{T_{\xx_1,F_1}})$. More precisely:

\begin{definition}
Suppose a Boolean network $F$ is of the form $ F(\xx_1,\xx_2)=(F_1(\xx_1),F_2(\xx_1,\xx_2))$. Then, we define the semidirect product 
\[
\mathcal{D}(F_1) \rtimes \mathcal{D}(F_2) := 
\sum_{\xx_1} T_{\xx_1,F_1} \times \mathcal{D}(F_2^{T_{\xx_1,F_1}}),
\]
where, in the last sum, $\xx_1$ varies over all elements of the domain of $F_1$.
\end{definition}

Note that since the first factor constrains which dynamics of $F_2$ are considered we cannot call this a product, and in analogy to group theory we call it a semidirect product. Another interpretation is that this product is a type of dot product of the dynamics of $F_1$ and the possible dynamics of  $F_2$. 

Now we prove the main theorem in this paper.
\begin{theorem}\label{thm:main}
Consider $F(\xx_1,\xx_2)=(F_1(\xx_1),F_2(\xx_1,\xx_2))$, then

\[
\mathcal{D}(F)  = \mathcal{D}(F_1) \rtimes \mathcal{D}(F_2).
\]
\end{theorem}
\begin{proof}
\begin{align*}
    \mathcal{D}(F) & = \sum_{\xx} T_{\xx,F}  & \text{ (definition of $\mathcal{D}(F)$) }\\
    & = \sum_{\xx_1,\xx_2} T_{(\xx_1,\xx_2),F} &  \text{ (enumerating the sum by $\xx_1$ and $\xx_2$ independently)} \\
         & = \sum_{\xx_1,\xx_2} T_{\xx_1,F_1}  \times T_{\xx_2,F_2}  & \text{ (Lemma \ref{lemma:prod}) }     \\
         & = \sum_{\xx_1} \sum_{\xx_2} T_{\xx_1,F_1}  \times T_{\xx_2,F_2}  & \text{(reordering of the sum)}\\
         & = \sum_{\xx_1} T_{\xx_1,F_1} \times \sum_{\xx_2}  T_{\xx_2,F_2}  & \text{ ( $\xx_2$ does not affect $F_1$)}\\
         & = \sum_{\xx_1}  T_{\xx_1,F_1} \times \mathcal{D}(F_2^{T_{\xx_1,F_1}}) & \text{  (definition of $\mathcal{D}(F_2^{T_{\xx_1,F_1}})$)} \\
         & = \mathcal{D}(F_1) \rtimes \mathcal{D}(F_2)
\end{align*} 
\end{proof}

Using Theorem.~\ref{thm:main}, we can now combine the dynamics of the modular components of a network (Fig.~\ref{fig:running_eg}de) to compute the global network dynamics. We show the complete and detailed computation in Fig.~\ref{fig:egF1F2}. Since the first two terms have identical second factors  (Fig.~\ref{fig:egF1F2}b), we can factor it. Then, the two terms resulting from the factoring (Fig.~\ref{fig:egF1F2}c) can be added, resulting in a shorter representation of the dynamics of the full system (Fig.~\ref{fig:egF1F2}d). Using the product from Definition~\ref{def:prod} results in the graphs in Fig.~\ref{fig:egF1F2}e. Finally, adding terms with overlapping elements yields a graph that is the state transition graph of the original system $F$, Fig.~\ref{fig:egF1F2}f.

\begin{figure}[h]
\centerline{ \hbox{ 
  \fig[scale=1.5]{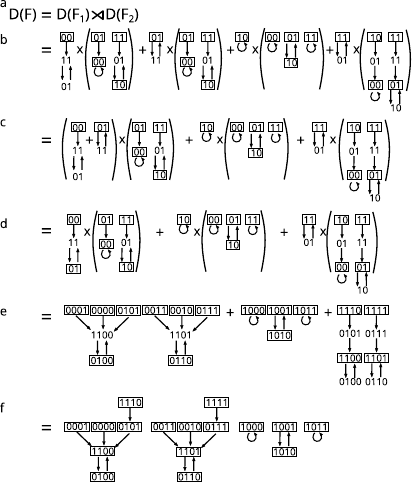}
}}
 \caption{ Illustrating Theorem~\ref{thm:main}. Using the semiring structure we can compute the dynamics of the full network using the dynamics of its modules.
}
 \label{fig:egF1F2}
\end{figure}

In practice, many of the factors $\mathcal{D}(F_2^{T_{\xx_1,F_1}})$ in the terms of $\mathcal{D}(F_1) \rtimes \mathcal{D}(F_2)$ can be identical for different $T_{\xx_1,F_1}$ (eg. see first and second terms in Fig.~\ref{fig:egF1F2}b). This can occur especially if the connection from $F_1$ to $F_2$ is sparse or canalizing. For such cases, we can factor  $\mathcal{D}(F_2^{T_{\xx_1,F_1}})$ from such terms (eg. Fig.~\ref{fig:egF1F2}c). This will result in $\mathcal{D}(F_1) \rtimes \mathcal{D}(F_2)$ having fewer terms, simplifying calculations.

\section{Conclusion}
In this paper, we have extended the mathematical formalism to represent and investigate the precise relationship between structure and dynamics of Boolean networks.  Specifically, we showed that the decomposition of a Boolean network into modules induces a corresponding decomposition of its entire state transition graph (i.e., all trajectories), not only its attractors. By endowing the space of possible dynamics with a semiring structure, we obtained a factorization of network dynamics in terms of the dynamics of its constituent modules. 

From a systems biology perspective, this result extends modular analysis beyond phenotypes, represented by attractors \cite{kadelka2023modularity,jarrah2010dynamics,carranza2025cattheory}, to the transient trajectories that underlie processes such as differentiation, reprogramming, signal transduction, and responses to perturbations. Because biological regulatory networks are often organized into relatively weakly coupled modules, our framework suggests new approaches for scalable simulation, analysis, and model reduction, where module-specific dynamics can be studied independently and subsequently combined to recover the dynamics of the full system.

The decomposition also provides a foundation for modular approaches to control and inference. Building on previous work on modular attractor control~\cite{murrugarra2025modular}, our trajectory-level decomposition naturally supports state-dependent and optimal control strategies. Likewise, it provides a foundation for a modular approach to reverse engineering from time-course data, in which individual modules are inferred separately and then assembled to reconstruct global network dynamics~\cite{wheeler2024modular}. More broadly, our results provide a mathematical foundation for studying how the modular organization of biological networks shapes their dynamical behavior.

\bibliographystyle{unsrt}
\bibliography{ref}

\end{document}